\pdfoutput=1

\documentclass[10pt,letterpaper]{article}

\usepackage{etex}

\usepackage{grffile}

\usepackage{tikz}
\usetikzlibrary{arrows}
\usetikzlibrary{plotmarks}
\usetikzlibrary{decorations.pathreplacing}

\usepackage{url, subfigure, graphicx, latexsym, multicol, multirow}
\usepackage{enumitem, xspace}

\graphicspath{{./}{./old-figures/}}

\usepackage{amsmath, amssymb, calc, printlen}

\usepackage{booktabs, ctable, footnote}
\usepackage[algoruled,vlined,algo2e]{algorithm2e}

\usepackage{mathptmx}

\usepackage{listings, verbatim}

\usepackage{amsthm}

\usepackage{hyperref}

\hypersetup{
  pdfauthor={Jinzhong Niu, Simon Parsons},
  pdftitle={Maximizing Matching in Double-sided Auctions},
  pdfstartview=FitH,
  breaklinks=true,
  bookmarksopen=true,
  bookmarksnumbered=true,
}

\newtheorem{definition}{Definition}
\newtheorem{lemma}{Lemma}
\newtheorem{corollary}{Corollary}
\newtheorem{theorem}{Theorem}

\newcommand{\mset}[1]{\textit{\textbf{#1}}\xspace}
\newcommand{\mseta}{\mset{A}}
\newcommand{\msetb}{\mset{B}}
\newcommand{\msetm}{\mset{M}}

\newcommand{\fn}[1]{\footnote{\scriptsize #1}}

\makeatletter
\@ifundefined{nstyle}{%

\usepackage{xifthen}

\newcommand{\acro}[1]{{\smaller \uppercase{#1}}}
\newcommand{\cda}{\acro{cda}\xspace}
\newcommand{\ch}{\acro{ch}\xspace}

\newcommand{\jcat}{\acro{jcat}\xspace}
\newcommand{\ttt}{\acro{tt}\xspace}
\newcommand{\tps}{\acro{ps}\xspace}
\newcommand{\tzic}{\acro{zi-c}\xspace}
\newcommand{\tgd}{\acro{gd}\xspace}

\newcommand{\ar}[1]{{\textsf{\smaller \uppercase{#1}}}\xspace}

\newcommand{\ars}[1]{{\smaller #1}\xspace}
\newcommand{\mv}{\ar{mv}}

\newcommand{\me}{\ar{me}}
\newcommand{\mt}{\ar{m$\theta$}}
\newcommand{\pair}[2]{\langle #1, #2\rangle}

\newcommand{\furl}[1]{{\footnotesize\url{#1}}\xspace}

\setlength{\heavyrulewidth }{0.1em}

\newcolumntype{+}{>{\global\let\currentrowstyle\relax}}
\newcolumntype{-}{>{\currentrowstyle}}

\DeclareMathOperator*{\argmax}{arg\,max}

\newlength{\clapwidth}
\clapwidth=20pt
\def\clap#1{\settowidth{\clapwidth}{#1} \hbox to 0.5\clapwidth{\hss#1\hss}}

\def\mathclap{\mathpalette\mathclapinternal}

\def\mathclapinternal#1#2{%
\clap{$\mathsurround=0pt#1{#2}$}}
}{%
}
\makeatother

\everymath{\displaystyle}

\begin{document}


\title{Maximizing Matching in Double-sided Auctions\fn{This is a full-length version of \cite{niu-aamas13-mv}.}}

\author{Jinzhong Niu\\
Department of Computer Science\\
The City College, The City University of New York\\
160 Convent Avenue, New York, NY 10031\\
\texttt{jniu@ccny.cuny.edu}
\and Simon Parsons\\
Department of Computer and Information Science\\
Brooklyn College, The City University of New York\\
2900 Bedford Avenue, Brooklyn, NY 11210\\
\texttt{parsons@sci.brooklyn.cuny.edu}
}

\maketitle


\begin{abstract}

In this paper, we introduce a novel, non-recursive, maximal matching algorithm for double auctions, which aims to maximize the amount of commodities to be traded. It differs from the usual equilibrium matching, which clears a market at the equilibrium price. We compare the two algorithms through experimental analyses, showing that the maximal matching algorithm is favored in scenarios where trading volume is a priority and that it may possibly improve allocative efficiency over equilibrium matching as well. A parameterized algorithm that incorporates both maximal matching and equilibrium matching as special cases is also presented to allow flexible control on how much to trade in a double auction.

\end{abstract}

\section{Introduction}
\label{sec:intro}

Double auctions are auctions where multiple buyers and multiple sellers trade~\cite{friedman-93-book-da.survey}. These auctions provide effective price discovery and high transaction throughput and, as a result, are widely used in financial markets \cite{schartz:francioni:weber:06:equity-trader-course} and automated control \cite{clearwater96mbc}, as well as in solving, for example, environmental problems \cite{rich-aer98-matching}. Clearly, matching buyers and sellers is a major issue in double auctions and providing effective matching mechanisms is the subject of this paper.

We start by assuming that each trader in an auction has a limit price, called its \emph{private value}, below which sellers will not sell and above which buyers will not buy. From the private values we can compute the \emph{demand}, the quantity of a commodity that buyers are prepared to purchase at each possible price, and the \emph{supply}, the quantity that sellers are prepared to sell at each possible price. If $price$ is plotted as a function of $quantity$ following the convention in economics, the \emph{demand curve} slopes downward and the \emph{supply curve} slopes upward, as shown in Figure~\ref{fig:ds-underlying}. When goods are indivisible, the curves have a stairwise shape.

Typically, there is some price at which the quantity demanded is equal to the quantity supplied. Graphically, this is the intersection of the supply and demand curves. The price is called the \emph{equilibrium price}, and the corresponding quantity of commodity is called the \emph{equilibrium quantity}. The most common solution to creating a matching algorithm is to clear the market at the price where the supply and the demand equal. This is called \emph{equilibrium matching}, and we denote a double auction with equilibrium matching as \me following the scheme in \cite{niu-jaamas10-cat}. The equilibrium price and equilibrium quantity are denoted as $p_0$ and $q_0$ respectively in Figure~\ref{fig:ds-underlying}. Traders whose private value is no less competitive than the equilibrium price are called \emph{intra-marginal} whereas the rest of the traders are called \emph{extra-marginal}. The intra-marginal traders correspond to the shaded area in  Figure~\ref{fig:ds-underlying}.

Now, the private values of traders are not publicly known in most practical scenarios. What is known instead are the prices that traders offer (called \emph{bids} for buyers, \emph{asks} for sellers, and \emph{shouts} when we don't distinguish), and self-interested traders will make offers away from their private values in order to make a profit. The prices and quantities that are offered also make a set of supply and demand curves, called the \emph{apparent} supply and demand curves, in contrast to the \emph{underlying} supply and demand based on traders' private values. Figure~\ref{fig:ds-apparent} shows that the apparent supply curve shifts up compared to the underlying supply curve in Figure~\ref{fig:ds-underlying}, while the apparent demand curve shifts down. Thus if we compute the apparent equilibrium, it may be some distance from the true equilibrium.

The challenge for developing matching algorithms is to have them function well given that they take the apparent supply and demand curves as input. In this paper, we introduce a novel matching algorithm that can potentially increase the quantity of goods to trade given the same supply and demand schedules, and analyze its performance in various scenarios.

\begin{figure*}[tb]
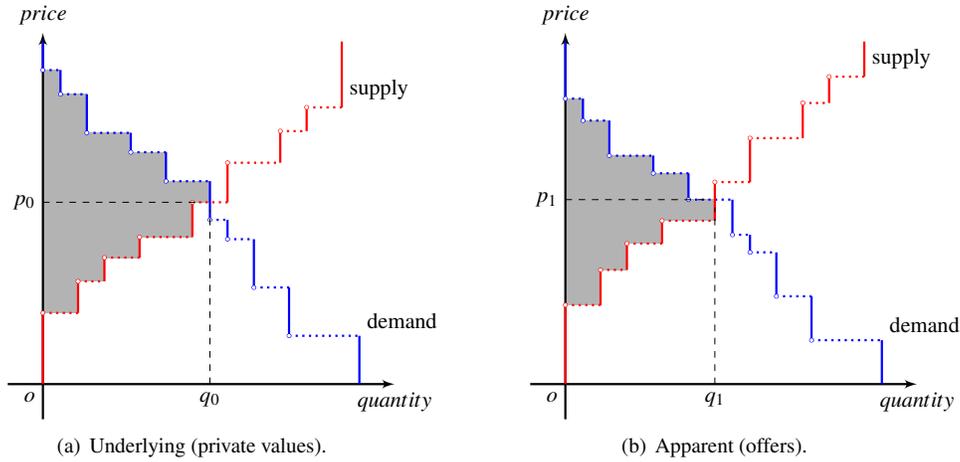

  \begin{center}
  \mbox{
    \subfigure[Underlying (private values).]{\label{fig:ds-underlying}\resizebox{.45\linewidth}{!}{\input{figs/ds-underlying-setup}\colorlet{d}{blue}
\colorlet{s}{red}
\colorlet{f}{white}
\colorlet{shade}{black!30}

\pgfdeclarelayer{background layer}
\pgfdeclarelayer{foreground layer}
\pgfsetlayers{background layer,main,foreground layer}

\begin{tikzpicture}[scale=.14]


\draw[line width=1pt, -latex'] (-4,0) -- (40,0) node[below] {$quantity$};
\draw[line width=1pt, -latex'] (0,-4) -- (0,40) node[above] {$price$};
\node[below left] (zero) at (0,0) {$o$};


\newcount\demand;
\demand=0;
\newcount\bidcount;
\bidcount=0;

\node (bid\the\bidcount) at (\demand,39) {};
\foreach \p/\q in {32.5/2, 30/3, 26/5, 24/4, 21/5, 17/2, 15/3, 10/4, 5/8} {
    \pgfmathparse{\p*\bidrate};
    \let\p\pgfmathresult;
    \node (xbid\the\bidcount) at (\demand,\p) {};
    \global\advance\bidcount by 1;
    \node (ybid\the\bidcount) at (\demand,\p) {};
    \global\advance\demand by \q;
    \node (bid\the\bidcount) at (\demand,\p) {};

    \ifthenelse{\dimtest{\p pt}{>}{\eqmy pt}}{%
        \begin{pgfonlayer}{background layer}
        \shade[top color=shade, bottom color=shade] (ybid\the\bidcount) rectangle (bid\the\bidcount |- 0,\eqmy);
        \end{pgfonlayer}
    }{%
    }
}

\node (xbid\the\bidcount) at (\demand,0) {};

\foreach \bidi in {0, ..., \bidcount} {
    \draw[opaque,line width=1pt,color=d] (bid\bidi.center) -- (xbid\bidi.center);
}
\foreach \bidi in {1, ..., \bidcount} {
    \draw[dotted,opaque,line width=1pt,color=d] (ybid\bidi.center) -- (bid\bidi.center);
    \draw plot[mark=o, mark options={color=d}, scale=3] coordinates {(ybid\bidi)};
    \draw plot[mark=*, mark options={color=f}, scale=2] coordinates {(ybid\bidi)};
}

\node[above right] at (bid\the\bidcount) {demand};


\newcount\supply;
\supply=0;
\newcount\askcount;
\askcount=0;

\node (ask\the\askcount) at (\supply,0) {};
\foreach \p/\q in {9/4, 13/3, 16/4, 18.6/6, 23/4, 28/6, 32/3, 35/4} {
    \pgfmathparse{\p*\askrate};
    \let\p\pgfmathresult;
    \node (xask\the\askcount) at (\supply,\p) {};
    \global\advance\askcount by 1;
    \node (yask\the\askcount) at (\supply,\p) {};
    \global\advance\supply by \q;
    \node (ask\the\askcount) at (\supply,\p) {};

    \ifthenelse{\dimtest{\p pt}{<}{\eqmy pt}}{%
        \begin{pgfonlayer}{background layer}
        \shade[top color=shade, bottom color=shade] (yask\the\askcount) rectangle (ask\the\askcount |- 0,\eqmy);
        \end{pgfonlayer}
    }{%
    }
}
\node (xask\the\askcount) at (\supply,39) {};

\foreach \aski in {0, ..., \askcount} {
    \draw[opaque,line width=1pt,color=s] (ask\aski.center) -- (xask\aski.center);
}

\foreach \aski in {1, ..., \askcount} {
    \draw[dotted,opaque,line width=1pt,color=s] (yask\aski.center) -- (ask\aski.center);
    \draw plot[mark=o, mark options={color=s}, scale=3] coordinates {(yask\aski)};
    \draw plot[mark=*, mark options={color=f}, scale=2] coordinates {(yask\aski)};
}

\node[above right] at (ask\the\askcount) {supply};


\node[left] (p0) at (0,\eqmy) {\eqmp};
\node[below] (q0) at (\eqmx,0) {\eqmq};
\node (eqm) at (\eqmx, \eqmy) {};

\draw[dashed, line width=0.25pt] (p0) -- (eqm) {};
\draw[dashed, line width=0.25pt] (q0) -- (eqm) {};

\end{tikzpicture}}}
    \subfigure[Apparent (offers).]{\label{fig:ds-apparent}\resizebox{.45\linewidth}{!}{\input{figs/ds-apparent-setup}\colorlet{d}{blue}
\colorlet{s}{red}
\colorlet{f}{white}
\colorlet{shade}{black!30}

\pgfdeclarelayer{background layer}
\pgfdeclarelayer{foreground layer}
\pgfsetlayers{background layer,main,foreground layer}

\begin{tikzpicture}[scale=.14]


\draw[line width=1pt, -latex'] (-4,0) -- (40,0) node[below] {$quantity$};
\draw[line width=1pt, -latex'] (0,-4) -- (0,40) node[above] {$price$};
\node[below left] (zero) at (0,0) {$o$};


\newcount\demand;
\demand=0;
\newcount\bidcount;
\bidcount=0;

\node (bid\the\bidcount) at (\demand,39) {};
\foreach \p/\q in {32.5/2, 30/3, 26/5, 24/4, 21/5, 17/2, 15/3, 10/4, 5/8} {
    \pgfmathparse{\p*\bidrate};
    \let\p\pgfmathresult;
    \node (xbid\the\bidcount) at (\demand,\p) {};
    \global\advance\bidcount by 1;
    \node (ybid\the\bidcount) at (\demand,\p) {};
    \global\advance\demand by \q;
    \node (bid\the\bidcount) at (\demand,\p) {};

    \ifthenelse{\dimtest{\p pt}{>}{\eqmy pt}}{%
        \begin{pgfonlayer}{background layer}
        \shade[top color=shade, bottom color=shade] (ybid\the\bidcount) rectangle (bid\the\bidcount |- 0,\eqmy);
        \end{pgfonlayer}
    }{%
    }
}

\node (xbid\the\bidcount) at (\demand,0) {};

\foreach \bidi in {0, ..., \bidcount} {
    \draw[opaque,line width=1pt,color=d] (bid\bidi.center) -- (xbid\bidi.center);
}
\foreach \bidi in {1, ..., \bidcount} {
    \draw[dotted,opaque,line width=1pt,color=d] (ybid\bidi.center) -- (bid\bidi.center);
    \draw plot[mark=o, mark options={color=d}, scale=3] coordinates {(ybid\bidi)};
    \draw plot[mark=*, mark options={color=f}, scale=2] coordinates {(ybid\bidi)};
}

\node[above right] at (bid\the\bidcount) {demand};


\newcount\supply;
\supply=0;
\newcount\askcount;
\askcount=0;

\node (ask\the\askcount) at (\supply,0) {};
\foreach \p/\q in {9/4, 13/3, 16/4, 18.6/6, 23/4, 28/6, 32/3, 35/4} {
    \pgfmathparse{\p*\askrate};
    \let\p\pgfmathresult;
    \node (xask\the\askcount) at (\supply,\p) {};
    \global\advance\askcount by 1;
    \node (yask\the\askcount) at (\supply,\p) {};
    \global\advance\supply by \q;
    \node (ask\the\askcount) at (\supply,\p) {};

    \ifthenelse{\dimtest{\p pt}{<}{\eqmy pt}}{%
        \begin{pgfonlayer}{background layer}
        \shade[top color=shade, bottom color=shade] (yask\the\askcount) rectangle (ask\the\askcount |- 0,\eqmy);
        \end{pgfonlayer}
    }{%
    }
}
\node (xask\the\askcount) at (\supply,39) {};

\foreach \aski in {0, ..., \askcount} {
    \draw[opaque,line width=1pt,color=s] (ask\aski.center) -- (xask\aski.center);
}

\foreach \aski in {1, ..., \askcount} {
    \draw[dotted,opaque,line width=1pt,color=s] (yask\aski.center) -- (ask\aski.center);
    \draw plot[mark=o, mark options={color=s}, scale=3] coordinates {(yask\aski)};
    \draw plot[mark=*, mark options={color=f}, scale=2] coordinates {(yask\aski)};
}

\node[above right] at (ask\the\askcount) {supply};


\node[left] (p0) at (0,\eqmy) {\eqmp};
\node[below] (q0) at (\eqmx,0) {\eqmq};
\node (eqm) at (\eqmx, \eqmy) {};

\draw[dashed, line width=0.25pt] (p0) -- (eqm) {};
\draw[dashed, line width=0.25pt] (q0) -- (eqm) {};

\end{tikzpicture}}}
  }
  \caption{Demand and supply schedules. Compared to the underlying supply and demand, the apparent supply shifts up and the apparent demand shifts down respectively as traders tend to offer higher prices to sell or lower prices to buy than their private values.}
  \label{fig:ds}
  \end{center}
\end{figure*}
%

\section{Background}
\label{sec:intro:me}

In a \me auction, shouts are matched when the market is cleared if and only if they are intra-marginal. If we denote the supply and the demand at price $p$ as $S(p)$ and $D(p)$, the trading volume with \me is:
\begin{equation}
Q_\ars{me} = \max_p\min\left[S(p), D(p)\right]
\label{equ:equilibrium-matching-quantity}
\end{equation}
and if all transactions are made at the price at which $Q_\ars{me}$ is achieved, what is called \emph{uniform pricing}, the transaction price can be defined as
\begin{equation}
P_\ars{me} = \argmax_p\min\left[S(p), D(p)\right]
\label{equ:equilibrium-matching-price}
\end{equation}
According to Wurman \textit{et al.} \cite{wurman-dss98-flexible.das}, $P_\ars{me}$ may also be represented by the $m\,th$ highest price among shouts or the $n\,th$ lowest price among shouts if there are respectively $m$ asks and $n$ bids in the market.

Prior work in experimental economics has shown that \me is efficient in common types of double auction. This result holds when the traders are human, and when they are software agents, even agents with little intelligence~\cite{cason-el92-call.market.eff,gode-sunder-93-jpe-zi,smith-62-jpe-competitive.market.behavior}. The notion of efficiency here is \emph{allocative efficiency}, denoted as $E_a$, a measure of how much social welfare is obtained through the auction. The \emph{equilibrium profit}, $P_e$, of an auction is\fn{For simplicity, we assume here that a trader has the same private value for each of the multiple units of commodity it wants to buy or sell. It may not be the case in reality.}
\begin{equation}
\label{equ:theoretical-profit}
P_{e}=\sum_{i}{|v_{i}-p_{0}| \cdot q_{i}}
\end{equation}
for all intra-marginal traders, where $p_{0}$ is the equilibrium price and $v_i$ is the private value of trader $i$ who can trade up to $q_i$ units of commodity at $p_{0}$ without a loss. This is the maximum total profit that could be gained in the auction. Graphically, $P_e$ is the size of the shaded area in Figure~\ref{fig:ds-underlying}. The \emph{actual overall profit}, $P_a$, of an auction is
\begin{equation}
\label{equ:actual-profit}
P_{a}=\sum_{j}{|v_{j}-p_{j}| \cdot q_{j}}
\end{equation}
where $p_j$ and $q_j$ are respectively the transaction price and quantity of a trade completed by trader $j$ and $v_j$ is the private value of trader $j$, where $j$ ranges over all traders who actually trade in the auction. $P_{a}$ is the counterpart of $P_{e}$ in Figure~\ref{fig:ds-apparent}.
Given these values, $E_a$ is the proportion of the equilibrium profit that is achieved in practice, $P_{a}/P_{e}$.
%

Allocative efficiency is one of the most common metrics for auctions, but there are other metrics that may carry more weight in certain scenarios. \emph{Trading volume} is one such metric. This measures the total amount of commodity transferred from sellers to buyers in an auction, and in markets where traders are charged a flat fee for each transaction it is clearly advantageous for the operators of those markets to maximize trading volume. (Although this may be at odds with the best interests of the traders, who benefit the most if allocative efficiency is maximized.) We use the term \emph{maximal-volume matching}, to describe markets that maximize trading volume.

With \me, all matched asks are priced no higher than the matched bids, and given this constraint \me actually realizes the maximal trading volume that is possible, as indicated in (\ref{equ:equilibrium-matching-quantity}) and as discussed in the use of a double auction to set pre-open prices on the Australian Stock Exchange \cite{ComertonForde-jfm06-call.auction}. Without this constraint on prices, there is the potential to further increase the trading volume. The idea is simply not to match between intra-marginal asks and intra-marginal bids as \me does, but to match extra-marginal asks with intra-marginal bids that are priced no lower than the asks, and match extra-marginal bids with intra-marginal asks that are priced no higher than the bids. This can potentially double the trading volume realized by \me.

This idea is not new. Indeed, we are aware of two pieces of prior work~\cite{rich-aer98-matching,zhao-ai2010-mm} on this topic. However neither of the algorithms proposed before are both computationally efficient and achieve the highest allocative efficiency possible while realizing the maximal trading volume. Rich \textit{et al.}~\cite{rich-aer98-matching} illustrated a scheme to do the matching, but depending upon the shape of the supply and demand curves in the market, the set of matching shouts that their algorithm produces may exclude more competitive shouts and include less competitive ones from the same side, failing on the property of \emph{fairness} that we will define below. Zhao \textit{et al.}~\cite{zhao-ai2010-mm}, on the other hand, presented an algorithm that correctly produces the set of matching shouts, but their algorithm is both computationally inefficient and obscure. Our main contribution is to introduce a computationally efficient, maximal-volume matching algorithm, but first we provide a formal model of the double auction market and use it to analyze the matching problem.

\section{Notation and Properties}
\label{sec:matching.set}


We denote the set of $n$ bids as $\msetb$ and the set of $m$ asks denoted as $\mseta$. For any $b$ in $\msetb$, $p(b)$ represents the price offered in the bid, and $q(b)$ represents the quantity of goods. Similarly we use $p(a)$ and $q(a)$ to represent the price and quantity in the ask.
For simplicity, we assume that all shouts are single-unit, e.g., $q(s)=1$, noting that a multi-unit shout can always be split into multiple single-unit ones with the same price, and we only consider the issue of matching at the moment that the auction clears rather than considering how supply and demand change over time.

We start with the idea of a matching set. This defines matches between bids and asks, such that each is matched at most once.
\begin{definition}
\label{def:matching-set}
Given a set of bids, $\msetb$, and a set of asks, $\mseta$, a \emph{matching set} between $\msetb$ and $\mseta$, denoted as $\msetm$ or more specifically $\msetm_{\msetb, \mseta}$, is a set of bid-ask pairs where: (1) for any $\pair{b}{a}$ in $\msetm$, $b$ is in $\msetb$, $a$ is in $\mseta$, $q(b)=q(a)$, and $p(b)\geq p(a)$; and (2) for any $b$ in $\msetb$ or any $a$ in $\mseta$, it appears at most in one of the bid-ask pairs in $\msetm$.
\end{definition}
The set of the bids that appear in $\msetm$ is denoted as $\msetb_\msetm$, and the set of the asks that appear in $\msetm$ is denoted as $\mseta_\msetm$.

%

From Definition~\ref{def:matching-set}, we know the number of bid-ask pairs in $\msetm_{\msetb, \mseta}$ equals the number of bids in $\msetb_\msetm$ and the number of asks in $\mseta_\msetm$, i.e.,
\[
|\msetm_{\msetb, \mseta}|=|\msetb_\msetm|=|\mseta_\msetm|.
\]
The trading volume that $\msetm_{\msetb, \mseta}$ achieves is denoted as
\[
\|\msetm_{\msetb,\mseta}\| = \|\msetb_\msetm\| = \|\mseta_\msetm\| = \sum_{\mathclap{b\in \msetb_\msetm}}{q(b)} = \sum_{\mathclap{a\in \mseta_\msetm}}{q(a)}.
\]
Since every shout offers to buy or sell one unit of goods, we have
\[
\|\msetm_{\msetb,\mseta}\|=|\msetm_{\msetb, \mseta}|.
\]
Matching sets are important because we can consider a \emph{matching algorithm} as determining how a market will clear by generating a matching set. So market clearing can be equated to generating a matching set.
%
%
\begin{lemma}
\label{lemma:demand-supply}
In a matching set, $\msetm_{\msetb,\mseta}$, the total quantity of goods that are offered to buy at any given price or higher is equal to or higher than the total quantity of goods that are offered to sell at the given price or higher; and the total quantity of goods that are offered to sell at any given price or lower is equal to or higher than the total quantity of goods that are offered to buy at the given price or lower.
\begin{align}
\sum_{\mathclap{b\in\msetb_\msetm,p(b)\geq p}}{q(b)} &\geq \sum_{\mathclap{a\in\mseta_\msetm,p(a)\geq p}}{q(a)} \label{equ:demand-higher-above} \\
\sum_{\mathclap{a\in\mseta_\msetm,p(a)\leq p}}{q(a)} &\geq \sum_{\mathclap{b\in\msetb_\msetm,p(b)\leq p}}{q(b)} \label{equ:supply-higher-below}
\end{align}
\end{lemma}
\begin{proof}
In any matching set there should always be enough bids that are priced no lower than a certain price to match asks that are priced in the same range, which is exactly (\ref{equ:demand-higher-above}). It is also true with the price range that excludes the given price itself, i.e.,
\begin{align}
\sum_{\mathclap{b\in\msetb_\msetm,p(b)> p}}{q(b)} &\geq \sum_{\mathclap{a\in\mseta_\msetm,p(a)> p}}{q(a)} \label{equ:demand-higher-above-properly}
\end{align}
Note that the only differences between (\ref{equ:demand-higher-above}) and (\ref{equ:demand-higher-above-properly}) are in the bottom parts of $\textstyle\sum$s, `$\geq$' in the former vs. `$>$' in the latter.

As, for any given $p$, the sum of the left-hand part of (\ref{equ:demand-higher-above-properly}) and the right-hand part of (\ref{equ:supply-higher-below}) equals $\|\msetb_\msetm\|$, i.e.,
\begin{equation}
\sum_{\mathclap{b\in\msetb_\msetm,p(b)> p}}{q(b)} + \sum_{\mathclap{b\in\msetb_\msetm,p(b)\leq p}}{q(b)} = \sum_{\mathclap{b\in\msetb_\msetm}}{q(b)} = \|\msetb_\msetm\| \label{equ:total-matched-demand}
\end{equation}
and the sum of the right-hand part of (\ref{equ:demand-higher-above-properly}) and the left-hand part of (\ref{equ:supply-higher-below}) equals $\|\mseta_\msetm\|$, i.e.,
\begin{equation}
\sum_{\mathclap{a\in\mseta_\msetm,p(a)> p}}{q(a)} + \sum_{\mathclap{a\in\mseta_\msetm,p(a)\leq p}}{q(a)} = \sum_{\mathclap{a\in\mseta_\msetm}}{q(a)} = \|\mseta_\msetm\| \label{equ:total-matched-supply}
\end{equation}
it is easy to derive from (\ref{equ:demand-higher-above-properly}) that (\ref{equ:supply-higher-below}) is also true.
\end{proof}
\begin{definition}
\label{def:fair}
A matching set between $\msetb$ and $\mseta$, $\msetm$, is \emph{fair} if (1) for any $b$ and $b'$ in $\msetb$ and $p(b)<p(b')$, if $b$ is in $\msetb_\msetm$, then $b'$ is in $\msetb_\msetm$ as well; and (2) for any $a$ and $a'$ in $\mseta$ and $p(a)>p(a')$, if $a$ is in $\mseta_\msetm$, then $a'$ is in $\mseta_\msetm$ as well.
\end{definition}
This says, roughly speaking, that a fair matching set includes the most competitive bids and asks no matter how much trading is defined by the matching set. Fairness is thus a desirable property of
the matching sets generated by a matching algorithm.
\begin{lemma}
\label{lemma:fair}
If $\msetm$ is a matching set between $\msetb$ and $\mseta$, then there exists a fair matching set between $\msetb$ and $\mseta$, $\mset{M'}$, such that $\msetm$ and $\mset{M'}$ achieve the same trading volume, i.e., $\|\msetm\|=\|\mset{M'}\|$.
\end{lemma}
\begin{proof}
We prove this lemma by directly constructing such a fair matching set, $\mset{M'}$, from $\msetm$. We start with a $\mset{M'}$ that is identical to $\msetm$.
If $\mset{M'}$ is not fair, there exists a bid-ask pair, $\pair{b}{a}$ in $\msetb_\mset{M'}$, such that there exists a $b'$ in $\msetb-\msetb_\mset{M'}$ and $p(b)<p(b')$, or a $a'$ in $\mseta-\mseta_\mset{M'}$ and $p(a)>p(a')$. In the case with $b'$, we can replace $b$ in $\pair{b}{a}$ and get $\pair{b'}{a}$. As $p(a)\leq p(b)<p(b')$, $\pair{b'}{a}$ is a valid match. In the symmetric case with $a'$, we replace $\pair{b}{a}$ in $\mset{M'}$ with $\pair{b}{a'}$. We repeatedly do this until $\mset{M'}$ is fair. As we do not change the quantity of goods to be matched through the process, the resulted $\mset{M'}$ achieves the same trading volume as $\msetm$.
\end{proof}
%

\begin{definition}
\label{def:reported-profit}
The \emph{reported overall profit} that a matching set, $\msetm$, achieves is the total trading profit, or the sum of bid-ask price difference for each unit of goods traded according to the matching set. It is denoted as
\begin{equation}
P_{\msetm} = \sum_{\mathclap{\pair{b}{a}\in\msetm}}{(p(b)-p(a))\cdot q(b)} \label{equ:report-profit}
\end{equation}
\end{definition}
Notice that the actual overall profit defined in (\ref{equ:actual-profit}) and the reported overall profit of a matching set defined here are different. The former is based on the private values of traders who trade according to the matching set while the latter is based on the prices that these traders offer. Both concern the same set of matches, the same set of traders, and the same quantity of goods to be reallocated, but only the reported overall profit is a metrics that can be actually calculated in real-world markets, where traders' private values are not available. Given a fixed trading volume, different matching sets may produce different reported overall profits, and generally speaking, the one producing higher reported overall profit is more likely to produce higher actual overall profit. This is because traders associated with more competitive private values are more likely to offer more competitive prices. This suggests that maximizing reported overall profit may potentially increase the allocative efficiency of the auction.

\begin{corollary}
\label{corollary:fair-maximizing-reported-profit}
Given a fixed trading volume, a fair matching set maximizes the reported overall profit.
\end{corollary}
The construction of a fair matching set in the proof of Lemma~\ref{lemma:fair} indicates that a fair matching set always includes the most competitive shouts from both sides and there is no way to further improve. Given a fixed trading volume, the most competitive shouts produce the maximal reported overall profit. So fairness guarantees that a more competitive shout has a higher priority to be included in the matching set --- a desirable property for individual traders --- and may potentially maximize the allocative efficiency given a fixed trading volume --- a desirable property for the market as a whole.

If fairness addresses the competitiveness of matched shouts over unmatched ones, orderliness addresses the relative competitiveness among matched shouts.
%
%
\begin{definition}
\label{def:orderly}
A matching set between $\msetb$ and $\mseta$, $\msetm$, is \emph{orderly} if, for any $\pair{b}{a}$ and $\pair{b'}{a'}$ in $\msetm$, $p(b)\leq p(b')$ when $p(a)\leq p(a')$, and vice versa.
\end{definition}
\begin{lemma}
\label{lemma:orderly}
Given a matching set between $\msetb$ and $\mseta$, $\msetm$, there exists an orderly matching set between $\msetb$ and $\mseta$, $\mset{M'}$, such that $\msetb_\msetm=\msetb_\mset{M'}$ and $\mseta_\msetm=\mseta_\mset{M'}$.
\end{lemma}
\begin{proof}
We prove this lemma by directly constructing such an orderly matching set, $\mset{M'}$, from $\msetm$, similar to the proof to Lemma~\ref{lemma:fair}. We start with a $\mset{M'}$ that is identical to $\msetm$.
If $\mset{M'}$ is not orderly, there exist two bid-ask pairs, $\pair{b}{a}$ and $\pair{b'}{a'}$ in $\mset{M'}$, such that $p(b)>p(b')$ and $p(a)\leq p(a')$, or $p(b)\leq p(b')$ and $p(a)>p(a')$. In the first case, as $p(a')\leq p(b')$, we know $p(a)\leq p(a')\leq p(b')<p(b)$. Thus we may switch $b$ and $b'$ in the two pairs, and replace the two pairs in $\mset{M'}$ with $\pair{b}{a'}$ and $\pair{b'}{a}$. The resulted $\mset{M'}$ is still a valid matching set. In the second case, a new matching set can be constructed similarly. We repeatedly do this until $\mset{M'}$ is orderly.
\end{proof}
Orderliness has no effect on changing the reported overall profit as the set of matched bids and asks remain the same. However it will affect how the actual overall profit is dispersed among traders who trade. In fact, since in an orderly matching set higher bids are paired up with higher asks, traders who placed these bids will trade at a higher transaction prices than if the common mid-point pricing rule\fn{That is that the transaction price for a match is set at the mid-point of the price range between the bid price and the ask price.} is used. An alternative way to pair up bids and asks into matches is to pair the highest bid with the lowest ask, and so on, leading to higher profit to traders who placed more competitive shouts --- the opposite to what is observed with an orderly matching set. Both the orderly matching set and the un-orderly alternative are possible when \me is used as any bid in the matching set can make a trade with any ask in the matching set. This is however not true when we try to increase the trading volume beyond that of \me. By doing so, one effect is that the profits of traders that placed the most competitive shouts is lowered. This may not be desirable for a marketplace that wants to keep those traders who are likely to be associated with the most competitive private values.
%
%

From Lemma~\ref{lemma:fair} and Lemma~\ref{lemma:orderly}, we may easily obtain the following corollaries:
\begin{corollary}
\label{corollary:fair-orderly}
Given a matching set between $\msetb$ and $\mseta$, $\msetm$, there exists a fair, orderly matching set between $\msetb$ and $\mseta$, $\mset{M'}$, such that $\|\msetm\|=\|\mset{M'}\|$.
\end{corollary}
\begin{corollary}
\label{corollary:fair-orderly-maximal}
Given a set of bids, $\msetb$, and a set of asks, $\mseta$, there exists a fair, orderly matching set between $\msetb$ and $\mseta$, $\-\msetm$, that achieves maximal trading volume.
\end{corollary}
Given a fixed trading volume, there may be multiple matching sets that are fair and orderly. This is because there may be two or more shouts from the same side having the same offered price. They may be matched in a matching set in different orders, or even some are the least competitive ones included in the matching set and the others are the most competitive ones excluded by the matching set, leading to different matching sets. Similarly there may be multiple fair, orderly matching sets that achieve the same maximal trading volume.

\section{Maximal-Volume Matching}
\label{sec:mv}

In this section, we present \mv, an algorithm that maximizes trading volume.
%
%
%
%
The \mv algorithm consists of two steps.
The first step calculates the maximal trading volume between the demand and supply schedules, and the second step determines the bid-ask pairs to form the matching set. The idea of the first step can be intuitively illustrated. Suppose the demand and supply schedules are as shown in Figure~\ref{fig:ds-apparent}. Then the supply schedule is flipped horizontally as depicted in Figure~\ref{fig:ds-flipped-shifted} before being shifted to the right towards the demand schedule until the two `touch'. The distance that the supply schedule moves is the minimal horizontal distance between the flipped supply curve and the demand curve, the very trading volume that the \mv algorithm computes as the maximal volume. As the distance between the demand curve and the supply curve at price $p$, shown along the horizontal dashed line in Figure~\ref{fig:ds-flipped-shifted}, is exactly the sum of the demand and the supply at $p$, the minimal horizontal distance or the maximal trading volume can be presented as
\begin{equation}
Q_\ars{mv} = \min_{p}\bigl(S(p) + D(p)\bigr)
\label{equ:matching-quantity}
\end{equation}
where $S(p)$ and $D(p)$ are respectively the supply and demand at price $p$.
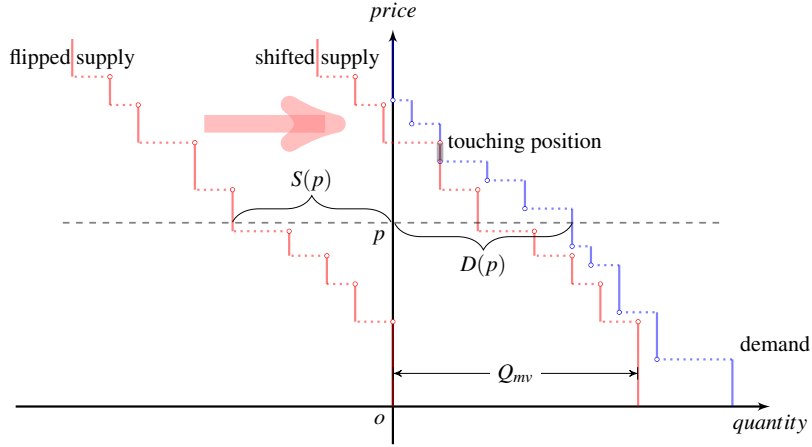
\begin{figure*}[tb]
  \begin{center}
  \resizebox{.75\linewidth}{!}{
  \mbox{%
    \colorlet{dc}{blue}
\colorlet{sc}{red}
\colorlet{fc}{white}
\colorlet{bc}{black}

\begin{tikzpicture}[scale=.15]


\draw[line width=1pt, -latex'] (-40,0) -- (40,0) node[below] {$quantity$};
\draw[line width=1pt, -latex'] (0,-4) -- (0,40) node[above] {$price$};
\node[below left] (zero) at (0,0) {$o$};


\newcount\demand;
\demand=0;
\newcount\bidcount;
\bidcount=0;

\node (bid\the\bidcount) at (\demand,39) {};
\foreach \p/\q in {32.5/2, 30/3, 26/5, 24/4, 21/5, 17/2, 15/3, 10/4, 5/8} {
    \node (xbid\the\bidcount) at (\demand,\p) {};
    \global\advance\bidcount by 1;
    \node (ybid\the\bidcount) at (\demand,\p) {};
    \global\advance\demand by \q;
    \node (bid\the\bidcount) at (\demand,\p) {};
}
\node (xbid\the\bidcount) at (\demand,0) {};

\foreach \bidi in {0, ..., \bidcount} {
    \draw[semitransparent,line width=1pt,color=dc] (bid\bidi.center) -- (xbid\bidi.center);
}
\foreach \bidi in {1, ..., \bidcount} {
    \draw[dotted,semitransparent,line width=1pt,color=dc] (ybid\bidi.center) -- (bid\bidi.center);
    \draw plot[mark=o, mark options={color=dc}, scale=3] coordinates {(ybid\bidi)};
    \draw plot[mark=*, mark options={color=fc}, scale=2] coordinates {(ybid\bidi)};
}

\node[above right] at (bid\the\bidcount) {demand};


\newcount\supply;
\supply=0;
\newcount\askcount;
\askcount=0;

\node (ask\the\askcount) at (\supply,0) {};
\foreach \p/\q in {9/4, 13/3, 16/4, 18.6/6, 23/4, 28/6, 32/3, 35/4} {
    \node (xask\the\askcount) at (\supply,\p) {};
    \global\advance\askcount by 1;
    \node (yask\the\askcount) at (\supply,\p) {};
    \global\advance\supply by -\q;
    \node (ask\the\askcount) at (\supply,\p) {};
}
\node (xask\the\askcount) at (\supply,39) {};

\foreach \aski in {0, ..., \askcount} {
    \draw[semitransparent,line width=1pt,color=sc] (ask\aski.center) -- (xask\aski.center);
}

\foreach \aski in {1, ..., \askcount} {
    \draw[dotted,semitransparent,line width=1pt,color=sc] (yask\aski.center) -- (ask\aski.center);
    \draw plot[mark=o, mark options={color=sc}, scale=3] coordinates {(yask\aski)};
    \draw plot[mark=*, mark options={color=fc}, scale=2] coordinates {(yask\aski)};
}

\node[above] at (ask\the\askcount) {flipped supply};


\newcount\supply;
\supply=26;
\newcount\askcount;
\askcount=0;

\node (ask\the\askcount) at (\supply,0) {};
\foreach \p/\q in {9/4, 13/3, 16/4, 18.6/6, 23/4, 28/6, 32/3, 35/4} {
    \node (xask\the\askcount) at (\supply,\p) {};
    \global\advance\askcount by 1;
    \node (yask\the\askcount) at (\supply,\p) {};
    \global\advance\supply by -\q;
    \node (ask\the\askcount) at (\supply,\p) {};
}
\node (xask\the\askcount) at (\supply,39) {};

\foreach \aski in {0, ..., \askcount} {
    \draw[semitransparent,line width=1pt,color=sc] (ask\aski.center) -- (xask\aski.center);
}

\foreach \aski in {1, ..., \askcount} {
    \draw[dotted,semitransparent,line width=1pt,color=sc] (yask\aski.center) -- (ask\aski.center);
    \draw plot[mark=o, mark options={color=sc}, scale=3] coordinates {(yask\aski)};
    \draw plot[mark=*, mark options={color=fc}, scale=2] coordinates {(yask\aski)};
}

\node[above] at (ask\the\askcount) {shifted supply};



\begin{scope}[>=stealth']
\draw[yshift=100pt,|<->|] (0,0) -- (26,0) node[fill=fc,pos=0.5]{$Q_{mv}$};
\end{scope}

\draw[semitransparent,line width=3pt, color=bc!60] (yask6.center) -- (xbid2.center);
\node[right,sloped] at (yask6) {touching position};

\draw[semitransparent,line width=8pt, color=sc!50, -stealth'] (-20,30) -- (-5,30);

\newcommand{\anyprice}{19.5}
\draw[semitransparent,thick, dashed,color=black] (-35,\anyprice) -- (35,\anyprice);
\node[below left] at (0,\anyprice) {$p$};

\draw [decorate,decoration={brace,amplitude=12pt},yshift=2pt]
(-17,\anyprice) -- (-0.2,\anyprice) node [bc,midway,yshift=19pt]
{$S(p)$};
\draw [decorate,decoration={brace,amplitude=12pt,mirror},yshift=-2pt]
(0.2,\anyprice) -- (19,\anyprice) node [bc,midway,yshift=-19pt]
{$D(p)$};
\end{tikzpicture}
  }
  }
  \caption{Supply schedule flipped horizontally and then shifted right towards demand schedule. The distance that the supply curve can move before touching the demand curve is exactly the maximal trading volume, $Q_\ars{mv}$.}
  \label{fig:ds-flipped-shifted}
  \end{center}
\end{figure*}

Algorithms~\ref{alg:mv} and \ref{alg:mvq} respectively show the \mv algorithm and the \ar{mv-getq} function that calculates the minimal horizontal distance between the flipped supply curve and the demand curve as shown in Figure~\ref{fig:ds-flipped-shifted} with bids and asks sorted respectively in the order of ascending price, as input. In Algorithm~\ref{alg:mvq}, horizontal distances between the two curves, recorded in $q$, are examined bottom-up by processing bids and asks in the order of ascending price, starting from price $0$. As the demand at price $0$, or the total demand in the market, is unknown at the beginning of the process, $q$ is not exactly the distance but the value of the distance minus the demand at price 0. This demand is calculated simultaneously and recorded in $q_d$. The minimal distance found through this process is adjusted eventually at the end of the \ar{mv-getq} function when $q_d$ is ready. Notice that in the \ar{mv-getq} function, all shouts are checked at most once. Asks that are priced beyond the highest bid price are not checked at all as they can only increase $q$, thus having no effect in obtaining the minimal $q$. With the trading volume known, constructing a matching set is simple as shown in Algorithm~\ref{alg:mv}.
\begin{algorithm2e}[!tb]
\caption{The \mv algorithm.}
\label{alg:mv}
\begin{small}
\LinesNumbered
\DontPrintSemicolon
\SetKwFunction{getnth}{getNth}
\SetKwFunction{len}{len}
\SetKwFunction{mvq}{MV-getQ}

\BlankLine
\KwIn{$Bids$ --- array of bids in the order of ascending price,
\mbox{$Asks$ --- array of asks in the order of ascending price}}
\KwOut{$\msetm$}
\BlankLine

\Begin{
\tcp{calculate the maximal trading volume.}
$q_\ars{mv} \leftarrow$ \mvq{$Bids$,\ $Asks$}\;\label{alg:mv-line1}
\BlankLine
\tcp{pair up the $q_\ars{mv}$ most competitive bids and asks in the order of ascending price}
$\msetm \leftarrow \varnothing$\;
$i \leftarrow 0$\;
\While{$i<q_\ars{mv}$}{\label{alg:mv-line2}
  $b \leftarrow$ \getnth{$Bids$,\ \len{$Bids$}$-q_\ars{mv}+i$}\;\label{alg:mv-line3}
  $a \leftarrow$ \getnth{$Asks$,\ $i$}\;
  $\msetm \leftarrow \msetm \cup \{\pair{b}{a}\}$\;
  $i++$\;
}}
\tcp{\len{}: returns the length of an array.}
\end{small}
\end{algorithm2e}
\begin{algorithm2e}[!tb]
\caption{The \ar{mv-getQ} function used in Algorithm~\ref{alg:mv}.}
\label{alg:mvq}
\begin{small}
\LinesNumbered
\DontPrintSemicolon
\SetKwFunction{poll}{poll}
\SetKwFunction{duplicate}{duplicate}
\SetKwFunction{reverse}{reverse}
\SetKwFunction{asqueue}{asQueue}
\SetKwFunction{minimum}{min}
\SetKw{kwis}{\textbf{is}}
\SetKw{kwisnot}{\textbf{is not}}
\SetKw{kwand}{\textbf{and}}

\BlankLine
\KwIn{\textit{omitted and same as in Algorithm~\ref{alg:mv}.}}
\KwOut{$q_{min}$}
\BlankLine

\Begin{

$q_{min} \leftarrow 0$\;
\BlankLine
\tcp{wrap up as queues}
$Bids \leftarrow$ \asqueue{$Bids$} \quad $Asks \leftarrow$ \asqueue{$Asks$}\;
\BlankLine
$a \leftarrow$ \poll{$Asks$}\;
\If{$a$ \kwisnot NULL}{
  $b \leftarrow$ \poll{$Bids$}\;
  \While{$b$ \kwisnot NULL \kwand $p(b) < p(a)$}{
    $b \leftarrow$ \poll{$Bids$}\;
  }
  \BlankLine
  \tcp{$q_d$: will be the demand at price $0$.}
  $q_d \leftarrow 0$\;
  \BlankLine
  \tcp{$q$: current horizontal distance between the flipped supply and the demand curves (minus the final value of $q_d$, which is unknown at present).}
  $q \leftarrow 0$\;
  \BlankLine
  \While{$b$ \kwisnot NULL}{
    \eIf{$a$ \kwisnot NULL \kwand $p(a)\leq p(b)$}{
      $q \leftarrow q\ +\ q(a)$\;
      $a \leftarrow$ \poll{$Asks$}\;
    }{
      $q \leftarrow q\ -\ q(b)$\;
      $q_{min} \leftarrow$ \minimum{$q_{min}$, $q$}\;
      $q_d \leftarrow q_d\ +\ q(b)$\;
      $b \leftarrow$ \poll{$Bids$}\;
    }
  }

  $q_{min} \leftarrow q_{min}\ +\ q_d$\;
}
}
\tcp{\asqueue{}: wrap up a given array as a read-only queue in the same order as items appear in the array.}
\end{small}
\end{algorithm2e}

Before we prove the soundness of the \mv algorithm, we examine the relationship between the trading volume that is achieved by an arbitrary matching set and the sum of demand and supply at any price. We have
\begin{lemma}
\label{lemma:upper-bound}
Given a matching set, $\msetm$, the trading volume achieved by $\msetm$ is no higher than the sum of the demand and the supply of the market at any given price, i.e.,
\[
\|\msetm\| \leq S(p) + D(p)
\]
for any price $p\geq 0$.
\end{lemma}
\begin{proof}
As any matching set has a fair, orderly counterpart that achieves the same trading volume according to Corollary~\ref{corollary:fair-orderly}, let us simply assume that $\msetm$ is fair and orderly.
Furthermore, because $\msetb_\msetm$ is a subset of $\msetb$, and $\mseta_\msetm$ is a subset of $\mseta$, we have
\begin{alignat*}{3}
S(p) &= \sum_{\mathclap{a\in\mseta,p(a)\leq p}}q(a) &&\geq \sum_{\mathclap{a\in\mseta_\msetm,p(a)\leq p}}q(a). \nonumber \\
D(p) &= \sum_{\mathclap{b\in\msetb,p(b)\geq p}}q(b) &&\geq \sum_{\mathclap{b\in\msetb_\msetm,p(b)\geq p}}q(b) \nonumber
\intertext{From (\ref{equ:demand-higher-above}) in Lemma~\ref{lemma:demand-supply}, we know}
D(p) &\geq \sum_{\mathclap{b\in\msetb_\msetm,p(b)\geq p}}q(b) &&\geq \sum_{\mathclap{a\in\mseta_\msetm,p(a)\geq p}}q(a) \nonumber \\
\ &\ &&= \|\msetm\| - \sum_{\mathclap{a\in\mseta_\msetm,p(a)<p}}q(a) \nonumber \\
\ &\ &&\geq \|\msetm\| - \sum_{\mathclap{a\in\mseta_\msetm,p(a)\leq p}}q(a) \nonumber \\
\ &\ &&\geq \|\msetm\| - S(p) \nonumber
\end{alignat*}
Adding $S(p)$ on both sides, we get
\[
\|\msetm\| \leq S(p) + D(p).
\]
%
\end{proof}
Now we can show that \mv meets the standards of fairness and orderliness that we defined above.
\begin{theorem}
\label{theorem:mv}
The \mv algorithm produces a fair, orderly matching set that maximizes the trading volume.
\end{theorem}
\begin{proof}
According to Lemma~\ref{lemma:upper-bound} and the arbitrariness of $p$, given any matching set $\msetm$, we know
\[
\|\msetm\| \leq \min_{p}\bigl(S(p) + D(p)\bigr),
\]
which gives an upper bound for the trading volume that could be achieved by a matching set. This upper bound is exactly $Q_\ars{mv}$ from (\ref{equ:matching-quantity}), the quantity calculated by the \ar{mv-getq} function in Algorithm~\ref{alg:mvq}.

This upper bound is also the part of supply in the shifted supply curve on the right-hand side of the $price$ axis in Figure~\ref{fig:ds-flipped-shifted}. Indeed, the bids and asks that are matched in the \mv algorithm form this part of supply and the part of demand that fall into the horizontal range labeled by $Q_\ars{mv}$ in Figure~\ref{fig:ds-flipped-shifted}. Because these bids and asks are processed in the order of ascending price and the shifted supply curve never goes beyond the demand curve vertically, each bid-ask pair produced by the \mv algorithm makes a valid match and all these pairs form a matching set that is fair and orderly and maximizes the trading volume.
\end{proof}

The \ar{mv-getq} function in Algorithm~\ref{alg:mvq} is a linear, one-pass scan of the shouts in the market, so the time complexity is $O(n+m)$. Thus the time complexity of the whole \mv algorithm is no more than generating sorted shouts as input, which is:
\[
O(\max(n,m)\cdot\log\max(n,m))
\]
In addition, the \mv algorithm consumes no more space than that needed for storing the sorted bids and asks, so its space complexity is
$O(n+m)$.
\mv is thus more efficient than \cite{zhao-ai2010-mm}'s approach to generating the same set of matching shouts.


So far, we have considered \mv as an alternative to the standard matching algorithm \me. However, we can also combine the two approaches. We do this in an attempt to  combine their advantages --- \me's ability to increase allocative efficiency with \mv's ability to increase trading volume. We call the combined policy \mt. It is parametric in the sense that its behavior is determined by a parameter  $\theta$ which specifies how to set the trading volume as a combination of $Q_\ars{mv}$, the trading volume computed by \ar{mv-getQ}, and $Q_\ars{me}$, the trading volume computed by the analogous process in equilibrium matching. Then, the trading volume set by this parametric policy is:
\begin{equation}
Q_\theta =
\begin{cases}
  (1+\theta)\cdot Q_\ars{me}  & \text{if $-1\leq \theta\leq 0$} \\
  \\
  (1-\theta)\cdot Q_\ars{me} + \theta\cdot Q_\ars{mv} & \text{if $0\leq \theta\leq 1$}
\end{cases}
\end{equation}
To obtain the \mt algorithm, we then substitute $Q_\theta$ for $Q_\ars{mv}$ in lines \ref{alg:mv-line1}, \ref{alg:mv-line2}, and \ref{alg:mv-line3} of
Algorithm~\ref{alg:mv}.

\mt represents a continuum of matching policies, including \me and \mv as special cases. When $\theta$ is
\begin{itemize}
\item[0:] $Q_\theta$ equals $Q_\ars{me}$, and \mt becomes \me;
\item[1:] $Q_\theta$ equals $Q_\ars{mv}$, and \mt is identical to \mv;
\item [-1:] $Q_\theta$ equals 0, and \mt does not match any offers.
\end{itemize}
In the following section we compare the performance of different values of $\theta$.

\section{Experimental Analysis}
\label{sec:experimental}

We evaluate the performance of maximal-volume matching by comparing a range of different auction mechanisms. Two of these are standard mechanisms --- the clearing house (\ch) and the continuous double auction (\cda), both of which employ the standard equilibrium matching policy \me. In the \ch, traders make offers, and at some predetermined closing time the matching algorithm is run, clearing the auction. In the \cda, the matching algorithm is run when every offer is made. The other auctions in our experiments are variations of the \ch which use various versions of \mt. Because these are essentially clearing houses, they run the matching algorithm at predetermined closing times, but the algorithm they use to do the matching varies from that in the standard \ch. We use three versions of \mt, with $\theta$ taking the values $-0.5$, $0.5$, and $1.0$ and we name these versions to reflect the value of $\theta$, calling them $\mt_{-0.5}$, $\mt_{0.5}$ and $\mv$.\fn{Of course, \ch is essentially the same as $\mt_{0}$.} When we need to distinguish this latter \mv from the \mv algorithm, we will call it ``the \mv market".

In our experiments, each market is populated by 20 traders, evenly split between sellers and buyers. Traders draw their private values from a uniform distribution between 50 and 150. In each experiment, traders all use the same strategy for deciding how to make offers. We ran experiments with five different strategies:
%
\emph{Truth Telling} (\ttt), which always bids at the private value;
\emph{Pure Simple} (\tps), which offers a price with a fixed profit level, $\delta$ \cite{zhan-jedc07-markups.in.cdas}\fn{``Pure simple'' is the name of this strategy in  \jcat \cite{niu-aamas08-jcat}, \cite{zhan-jedc07-markups.in.cdas} just calls it a ``markup strategy''.};
\emph{Zero Intelligence with Constraint} (\tzic), which randomly picks a price to offer but avoids trading at loss \cite{gode-sunder-93-jpe-zi}; and
\emph{Gjerstad and Dickhaut} (\tgd), which estimates the probability of an offer being accepted from the distribution of past offers, and chooses the offer which maximizes its expected utility \cite{gjerstad-dickhaut-98-geb-gd}.
%
Due to the parameter in \tps, we consider multiple versions of \tps, representing traders that are greedy at different levels, with $\delta$ taking the values $5$, $10$, $15$ or $20$. (That is we ran experiments for each of these different markups, and for each we ran experiments where all traders used the same markup). Note that when $\delta$ is 0, \tps is exactly \ttt.
%
%
We ran each experiment 100 times,  and experiments were run using \jcat, open-source software for the simulation of auction mechanisms \cite{niu-aamas08-jcat}. In \jcat, as is common in agent-based market simulations, simulations are broken into \emph{days} and \emph{rounds}. A day is intended to model a day's activity in the market, traders come to the market each day with new trade entitlements (money for buyers, goods for sellers) but can recall the progress of the previous day's trading, so learning is possible across days. A round is an opportunity for a trader to place a new shout or to modify their existing shout, and multiple rounds give traders a chance to watch prices evolve.

The first set of experiments we ran were intended to establish a performance baseline over a single day's trading. We ran 35 of these baseline experiments in total (each repeated 100 times), one for each combination of the five auction mechanisms (including the three \mt mechanisms) and the seven types of trading strategy (including the four \tps strategies). Each experiment ran for one trading day and each day had a single round. (Since most of our trading strategies do not learn, limiting them to a single round and a single day is not a big handicap.)

The results of the baseline experiments are given in Figure~\ref{fig:baseline}, which shows both the trading volume and the allocative efficiency of the markets.
\begin{figure*}[tb]
\begin{center}
  \mbox{
    \subfigure[Trading Volume.]{\label{fig:baseline-trading-volume}\includegraphics[width=.43\linewidth]{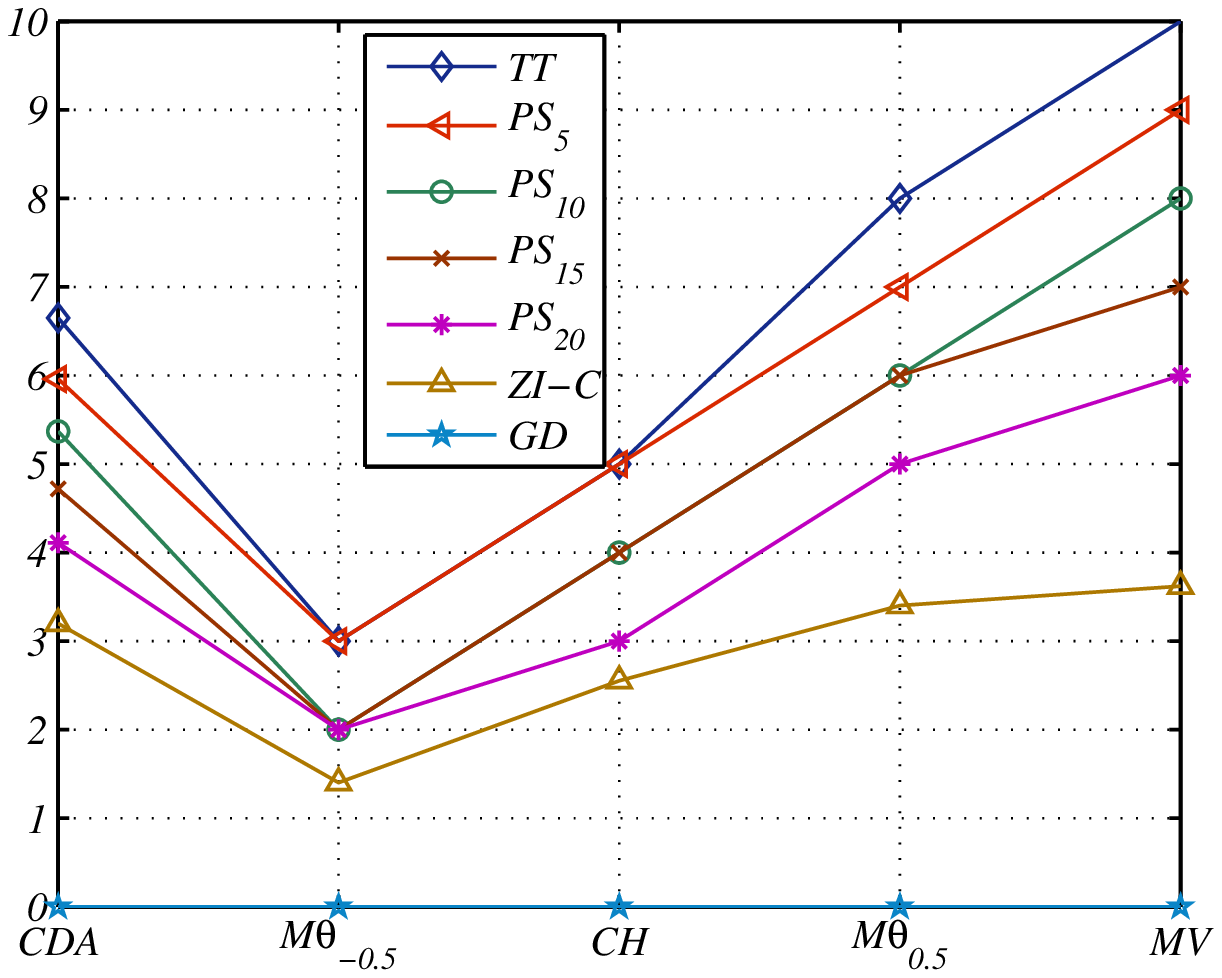}}
    \quad
    \subfigure[Allocative Efficiency.]{\label{fig:baseline-eff}\includegraphics[width=.43\linewidth]{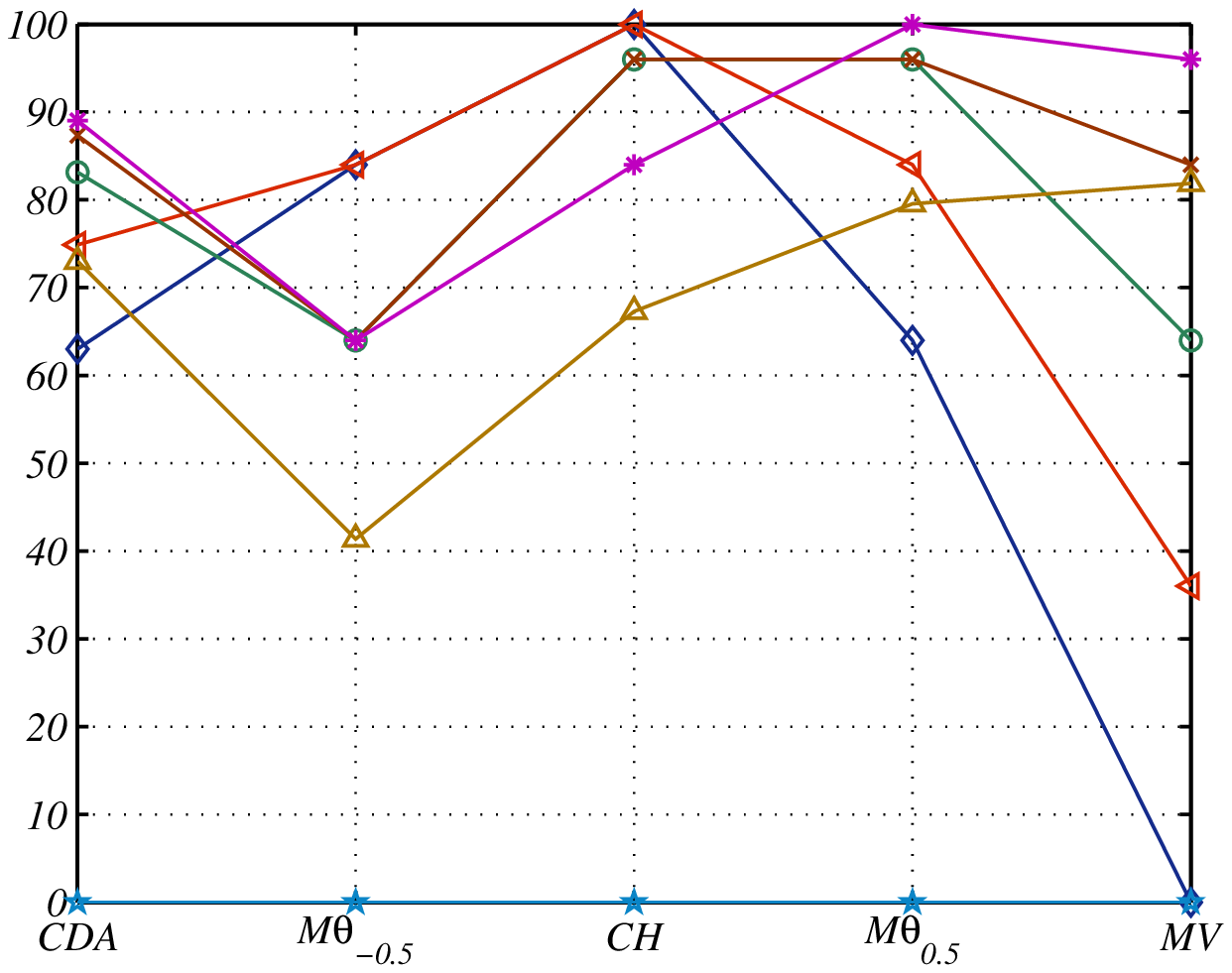}}
  }
\caption{Markets running a single day and a single round for that day.}\label{fig:baseline}
\end{center}
\end{figure*}

The main thing to note about the results in Figure~\ref{fig:baseline-trading-volume} is the steady increase in trading volume as we move from $\mt_{-0.5}$ to $\mt_{1.0}$ (\ch, as noted before, is $\mt_{0}$, and the \mv market is $\mt_{1.0}$). Thus, within the structure of the traditional clearing house, the maximal matching algorithm \mv does exactly what it is supposed to --- it increases trading volume. In terms of trading volume, the \cda performs better than the \ch for all trading strategies, despite the fact that it is using the same equilibrium matching algorithm \me. This is because the market can clear after each offer is received --- the rest is to randomize the set of possible matches, throwing up the chance of a match between an intra-marginal and an extra-marginal trader (that this is the effect of the \cda is not a new observation, but to our knowledge this is the first time that it has been related to trading volume). Also worth noting is that the trading volume decreases monotonically with the greediness of the traders --- the markup  that they apply to their private value in establishing their offers.\fn{Since \tzic picks a random markup that can be 100\% of the private value, on average \tzic is greedier than the greediest \tps strategy.} The effect of the \mv algorithm component in \mt is to increase the number of traders that can be matched, while the effect of the markups is to decrease this number. The bigger the markup, the more it offsets the gains made by the \mv component (again the effect of the margins is not a new observation, having been made in \cite{zhan-jedc07-markups.in.cdas}, but it hasn't before been related to trading volume).

In showing that the \mv algorithm does its job in increasing trading volume, Figure~\ref{fig:baseline-trading-volume} stresses the benefit of the approach. Figure~\ref{fig:baseline-eff} gives the cost. Here we see that as $\theta$ increases from $-0.5$ to $1$, allocative efficiency first increases, and then, for less greedy strategies, falls dramatically. However, for greedier traders, this drop does not offset the initial increase. If the markup is too small, then the maximal-volume matching includes lots of unprofitable trades, and there is little actual profit to generate high efficiency. Big markups, however, reduce the number of trades that are allowable (exactly the effect discussed by \cite{zhan-jedc07-markups.in.cdas} again), but the ones that are ruled out are the least profitable, so efficiency is higher. We should also point out that $\mt_{0.5}$ obtains higher allocative efficiency than the \cda for all trading strategies, indicating that some intermediate matching policy between \me and \mv allows the clearing house-style auction to beat the \cda in terms of both volume and allocative efficiency.

Now, these baseline experiments are somewhat artificial since they limit every trader to making a single offer. (The lack of learning opportunities explains why \tgd fails to make any trade in the results of Figure~\ref{fig:baseline}.) They are helpful in comparing the various matching policies based on a snapshot of supply and demand, the scenario that we considered in designing \mv. However, to understand how the \mv algorithm would help in more realistic markets we ran a second set of experiments. In these experiments we increased the number of chances each trader has to make offers. We still looked at just a single trading day, but looked at performance as the number of rounds of bidding was increased from 1 to 10. Figure~\ref{fig:multirounds} shows the results for the \cda,  \ch, and \mv markets.
\begin{figure*}[tb]
\begin{center}
    \subfigure[Trading volume of \cda.]{\label{fig:multirounds-cda-trading-volume}\includegraphics[width=.32\linewidth]{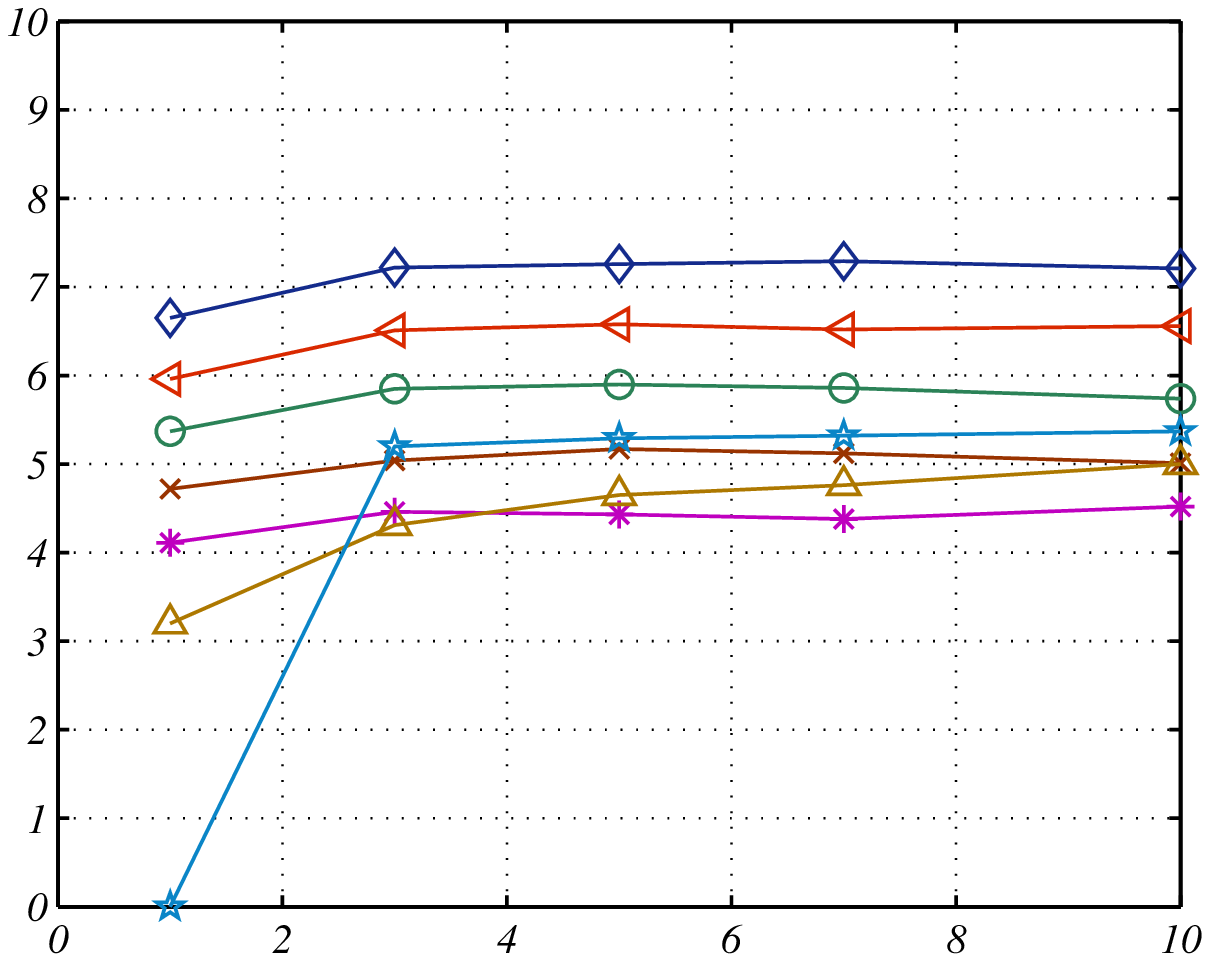}}
    \subfigure[Trading volume of \ch.]{\label{fig:multirounds-ch-trading-volume}\includegraphics[width=.32\linewidth]{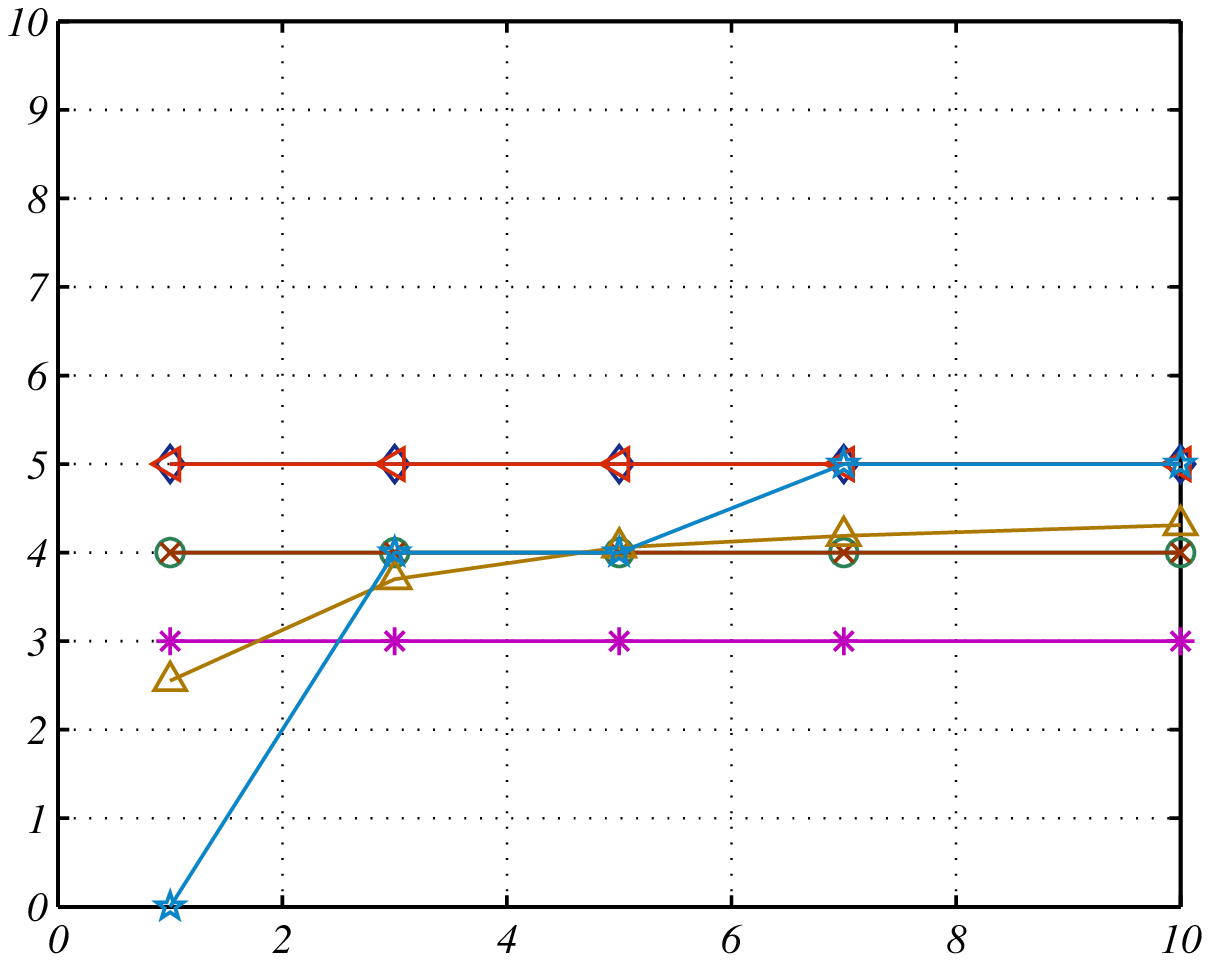}}
    \subfigure[Trading volume of \mv.]{\label{fig:multirounds-mt1.0-trading-volume}\includegraphics[width=.32\linewidth]{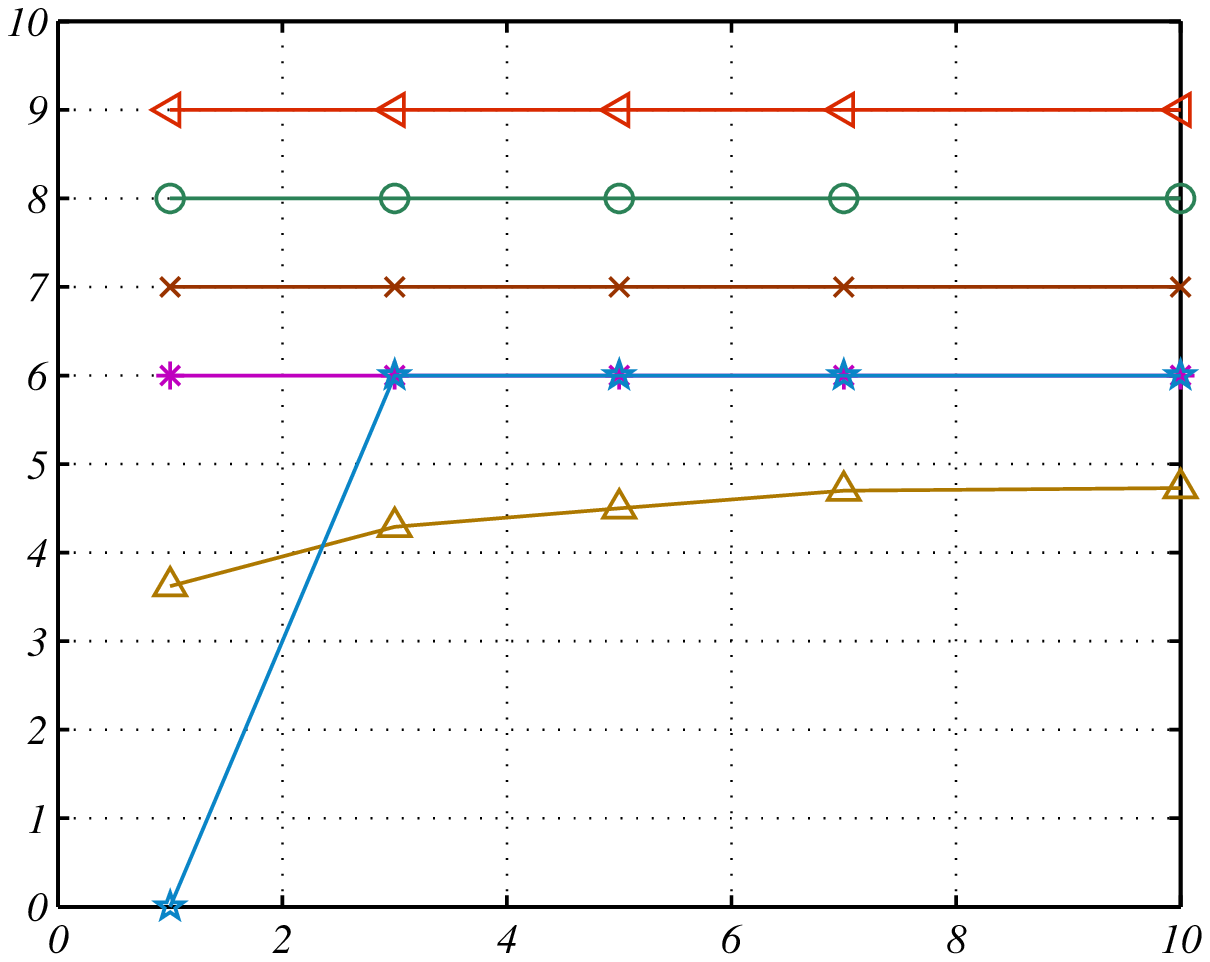}}
  \mbox{
    \subfigure[Allocative Efficiency of \cda.]{\label{fig:multirounds-cda-eff}\includegraphics[width=.32\linewidth]{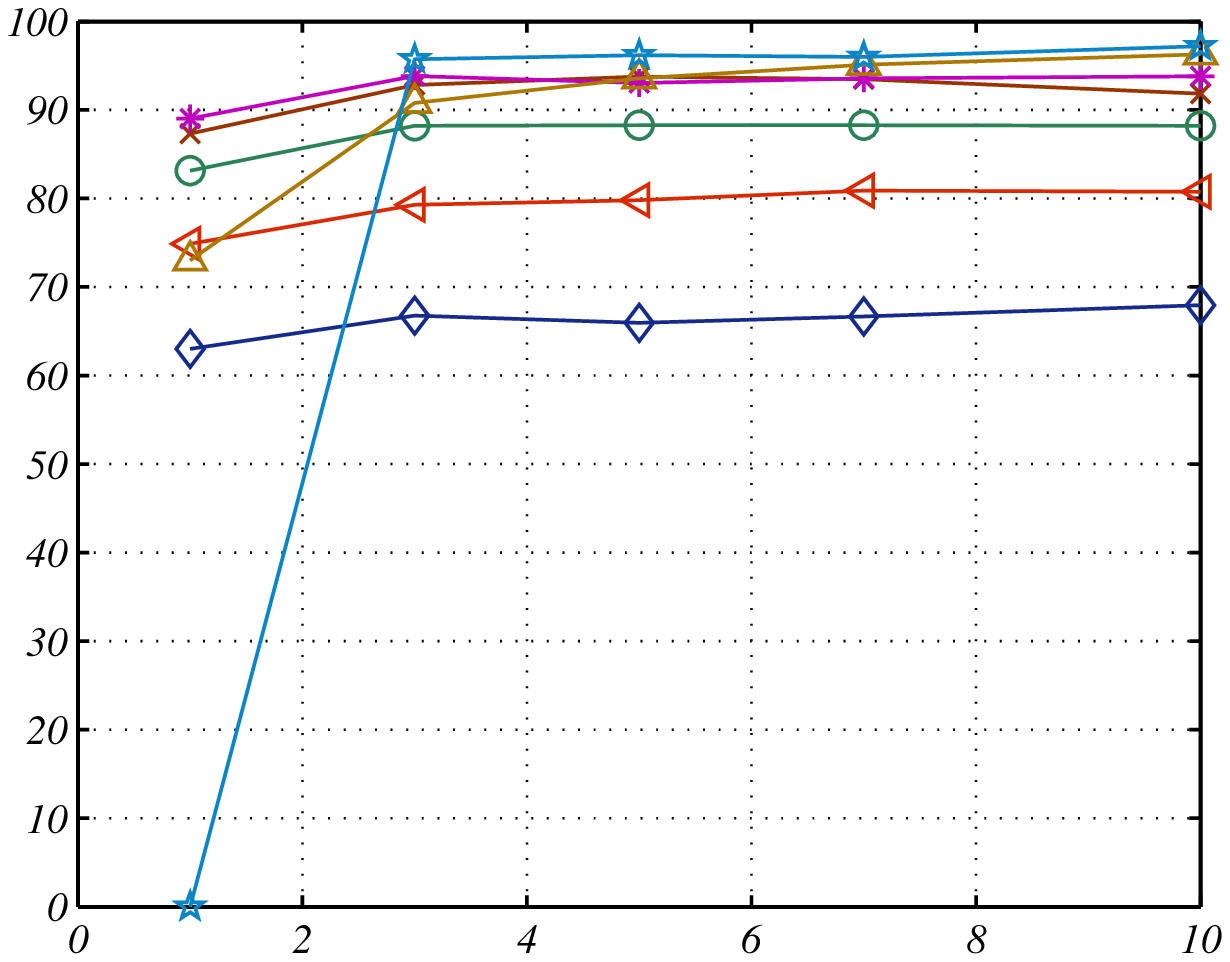}}
    \subfigure[Allocative Efficiency of \ch.]{\label{fig:multirounds-ch-eff}\includegraphics[width=.32\linewidth]{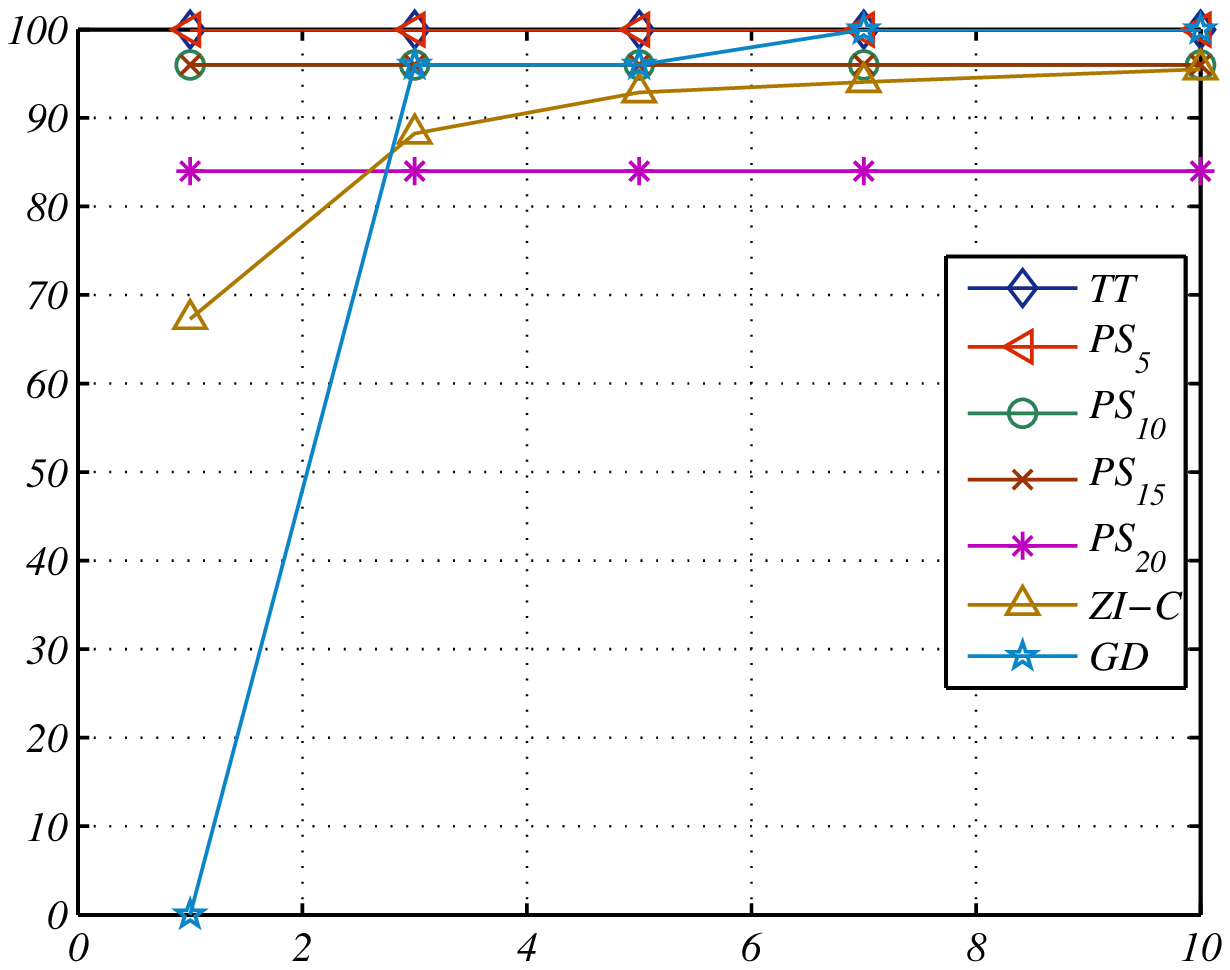}}
    \subfigure[Allocative Efficiency of \mv.]{\label{fig:multirounds-mt1.0-eff}\includegraphics[width=.32\linewidth]{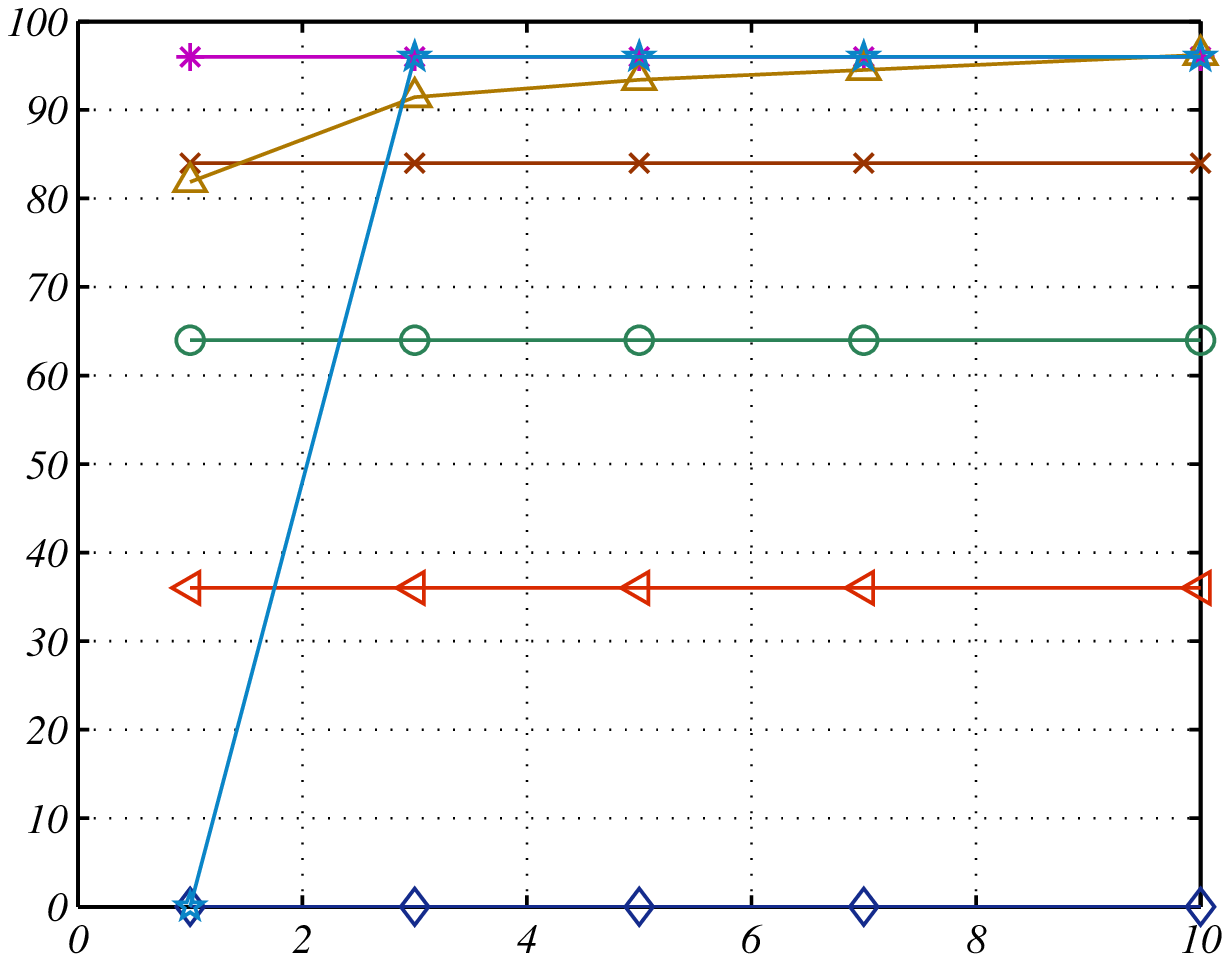}}
  }
\caption{Markets running multiple rounds per day. The $x$ axis gives the number of rounds and the $y$ axis gives the metrics of concern. }\label{fig:multirounds}
\end{center}
\end{figure*}

The most obvious change, both in trading volume and efficiency, is that \tgd, now given the chance to learn, greatly improves in performance. This of course is what we would expect from an effective trading strategy, and \tgd is one of the highest performing strategies known. We also see the well known effect of \tzic improving its performance over time despite the fact that it does not learn. (This was first reported in \cite{gode-sunder-93-jpe-zi} and explained by \cite{cliff-97-tr-zip} --- traders who do not trade in early rounds are given more chances later to pick a more competitive offer and be matched, and the shape of the supply and demand curves means that traders who trade later are more likely to be more profitable.)

However, the most important result from our perspective, is that the \mv market generates higher trading volumes than the \ch and the \cda across all the trader types that we tested. So, it can be argued that the \mv algorithm, despite being designed by considering a single round of offers and unchanging supply and demand curves, holds up when the supply and demand curves change as a result of trades made in previous rounds (since we are looking at the situation within one day, trade entitlements are not renewed). Furthermore, for certain traders --- the zero intelligence traders \tzic and the sophisticated \tgd traders --- the \mv market can achieve the same high levels of allocative efficiency that are obtainable in the \cda and the \ch. Together, these two results suggest that the \mv algorithm has the potential to be adopted in real-world markets. This is because it balances the objectives of the main parties in those markets, providing high efficiency, which is what traders want, while pushing up volume and thus allowing those who operate the markets to increase their profits.

\section{Summary}
\label{sec:summary}

This paper describes our work considering double auction markets from the perspective of maximizing the volume of trade. We developed some properties that are desirable for an algorithm that provides maximal-volume matching, and then developed an algorithm, that we call \mv, which was proved to have exactly those properties. We then evaluated \mv experimentally, showing that it does indeed increase trading volume, and that, as predicted, this increase comes at a cost in terms of allocative efficiency. However, further investigation also showed that this cost is largely overcome in markets that involve homogeneous populations of traders using two common trading strategies, \tzic and \tgd. Our approach outperforms two previous approaches to maximal-volume matching, being more efficient than \cite{zhao-ai2010-mm}, and, unlike \cite{rich-aer98-matching}, having the property of fairness.

\end{document}